\documentclass{ifacconf}
\usepackage[utf8]{inputenc}   
\usepackage{newunicodechar}
\newunicodechar{−}{-}   

\usepackage{graphicx}      
\usepackage{natbib}        
\usepackage{amssymb,amsfonts}
\usepackage{commath}
\usepackage{amsmath}
\usepackage{bm}
\usepackage{caption}
\usepackage{algpseudocode}
\usepackage{algorithm}
\usepackage{graphicx}
\usepackage{textcomp}
\usepackage{xcolor}
\usepackage{multirow}
\usepackage{multicol}
\usepackage{float}
\usepackage{booktabs}
\usepackage{stackengine}
\usepackage{etoolbox}
\usepackage{mathtools}
\usepackage{subcaption}
\usepackage{tabularx}
\usepackage{cases}

\newtheorem{theorem}{Theorem}
\newtheorem{definition}{Definition}

\newtheorem{corollary}{Corollary}
\newtheorem{lemma}{Lemma}

\newenvironment{proof}{%
	\par\noindent\textbf{Proof.}\ }{%
	\hfill$\square$\par
}

\newcommand{\Rmin}{R_{\min}}
\newcommand{\Lmin}{L_{\mathrm{min}}}
\newcommand{\Lmax}{L_{\mathrm{max}}}
\newcommand{\Lc}{L_0^{\mathrm{c}}}
\newcommand{\Lv}{L_0^{\mathrm{v}}}

\newcommand{\RNum}[1]{\uppercase\expandafter{\romannumeral#1\relax}}
\DeclareMathOperator{\sign}{sign}

\graphicspath{{figsIFAC2026/}}
\begin{document}
	\begin{frontmatter}
		
		\title{Variable $L_0$ Guidance Strategy: Enlarged Operational Envelope and Path-Following
			\thanksref{sponser}} 
		
		
		\thanks[sponser]{This research has been supported by the Italian Ministry of Enterprises and Made in Italy under grant "4DDS - 4D Drone Swarms" and by  FCT/MCTES (PIDDAC), through project 2022.02801.PTDC-UPWIND-ATOL (https://doi.org/10.54499/2022.02801.PTDC) and grant 2021.07346.BD.}
		

\author[First]{Amit Shivam,} 
\author[First]{Manuel C.R.M. Fernandes,} 
\author[First]{Fernando A.C.C. Fontes,}
\author[Second]{Lorenzo Fagiano}

\address[First]{Department of Electrical and Computer Engineering, Faculdade de Engenharia, Universidade do Porto, Porto, Portugal \newline (e-mail: amitshivam@alum.iisc.ac.in,\{mcrmf, faf\}@fe.up.pt).}
\address[Second]{Department of Electronics, Information and Bioengineering, Polytechnic University of Milan, Italy \newline (e-mail: lorenzo.fagiano@polimi.it)}
\begin{abstract}                

   This paper presents a geometric and theoretical study of an exponentially varying look-ahead parameter for UAV path-following guidance. Conventional guidance laws with a fixed look-ahead distance often drive the vehicle into turn-rate saturation when the heading or cross-track error is large, leading to constrained maneuvers and higher control effort. The proposed variable $L_0$
 strategy reshapes the look-ahead profile so that the guidance command adapts to the evolving tracking error geometry. A detailed investigation shows that this adaptation significantly enlarges the region in which the commanded turn rate remains unsaturated, allowing the vehicle to operate smoothly over a broader range of error conditions. For representative settings, the unsaturated operational envelope increases by more than 70\% relative to the constant $L_0$
 formulation. These geometric insights translate to smoother trajectories, earlier recovery from saturation, and reduced control demand. Simulation studies on straight-line and elliptical paths demonstrate the merits of the variable look-ahead strategy, highlighting its control-efficient and reliable path-following performance.
\end{abstract}


\end{frontmatter}
\section{Introduction}

The rapid expansion of uncrewed aerial vehicles (UAVs) in security, environmental monitoring, medical logistics, and civilian automation has increased the need for reliable path-following guidance systems. These applications often require UAV to autonomously follow a desired geometric path while considering the vehicle’s kinematic and actuator constraints. Although several guidance laws provide stable convergence, most rely on fixed look-ahead distances or curvature parameters that do not adapt when the vehicle experiences large heading deviations or significant lateral offsets. Under such conditions, the commanded turn rate often saturates, leading to constrained maneuvers and elevated control effort. This underscores the importance of guidance strategies with a large operational envelope, that is, the range of tracking errors for which the controller can operate without saturating the actuators. Expanding this envelope directly improves tracking smoothness, energy usage, and actuator-feasible behavior. However, despite progress in UAV guidance design, the underlying geometric structure of this operational envelope and the factors governing its boundary  remains largely unexplored.

Several path-following strategies have been developed to address some of these issues. Waypoint-based navigation \cite{osborne2005waypoint,hota2014curvature} and control-theoretic approaches such as MPC \cite{Alessandretti2013}, LQR \cite{lee2010LQR}, and sliding-mode control \cite{shashiranjanAST} provide systematic frameworks for trajectory tracking. Vector-field methods \cite{frew2008coordinated,pothen2017curvature,goncalves2009artificial,kapitanyuk2017guiding,nelson2007,amitICUAS,shivam2023} and virtual-target or $L_1$-type guidance laws ($L_1$ is the distance between the vehicle and the virtual target on the desired path) \cite{breivik2005principles,park_new_nodate,park_performance_2007,curry2013l+,stastny2018l1} remain popular due to their computational simplicity and robustness. Variable look-ahead formulations have been explored for straight-line \cite{lekkas2012time} and spline-based paths \cite{lekkas2014integral}, while trajectory-shaping guidance laws have been proposed for curved path following \cite{ratnoo2015path}. In the later course of work, ~\cite{thakar2017tangential} developed lead-angle–based guidance logic, which further improves transient performance and settling times in comparison to \cite{park_performance_2007,ratnoo2015path}. The constant $L_0$
 guidance framework introduced in \cite{silva_path-following_nodate} enlarged the domain of attraction relative to the classical $L_1$ 
 setting and was later analyzed for curved paths~\cite{fernandes_path-following_2024}  when integrated with predictive control schemes~\cite{fernandes_model_2020}. Comparative studies such as \cite{l0l1_energies} further highlight the practical need for these variants.

 Although these methods provide meaningful improvements, the existing literature offers limited theoretical insight into how the look-ahead parameter—whether fixed or adaptive—shapes the boundary between saturated and unsaturated operation. In particular, a clear geometric characterization of the operational envelope, and the factors that determine the onset of saturation, lacks in the literature. 

 Motivated by this gap, the present work offers a geometric and theoretical study of an exponentially varying $L_0$
 profile and examines its impact on the path-following behavior of a UAV kinematic model. Building on preliminary observations reported in \cite{FERNANDES2025128}, we show that the proposed variable $L_0$
 formulation significantly enlarges the range of tracking conditions under which the guidance command remains unsaturated. This expanded operational envelope enables the controller to operate smoothly across a broader set of error scenarios, leading to earlier recovery from saturation, smoother trajectories, and reduced control effort. The geometric analysis provides explicit criteria that characterizes when saturation occurs and when it can be avoided. Simulation studies on straight-line and elliptical paths further illustrate the merits of the variable 
$L_0$ guidance strategy.

\par 

This paper is organized as follows. Section~\ref{sec:pf} discusses the problem formulation. The proposed varying $L_0$ guidance law is presented in Section~\ref{sec:proposedguidance}.  
Section~\ref{sec:simulations} discusses simulation results, and Section~\ref{sec:Conclusions} provides the concluding remarks.

\section{Problem Formulation}\label{sec:pf}
Consider a path-following scenario as illustrated in Fig.~\ref{fig: Path-following geometry}.
Let the UAV state be $(x,y)\in\mathbb{R}^2$ with heading angle $\psi\in(-\pi,\pi]$ and constant speed $V_u>0$.
The objective here is to develop a prospective virtual target-based guidance law that governs the UAV heading angle while pursuing a virtual target $T$ placed ahead along the path, eventually leading the UAV onto the desired path. 

\begin{figure}[t]
\centering
\includegraphics[width=0.85\linewidth]{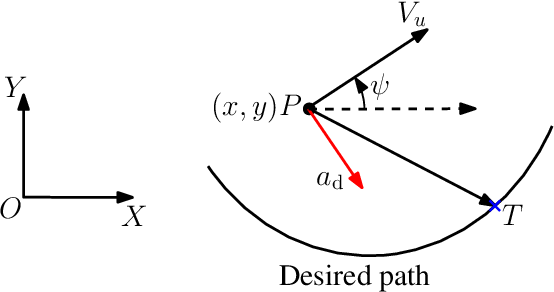}
\caption{Path-following geometry.}
\label{fig: Path-following geometry}
\end{figure}
The kinematic equations of motion are described by
\begin{equation}
\label{eq:kinematics}
\begin{aligned}
\dot{x} &= V_u \cos\psi, \qquad
\dot{y}  = V_u \sin\psi, \qquad
\dot{\psi} = \frac{a_{\mathrm d}}{V_u}.
\end{aligned}
\end{equation}
Here, $a_\mathrm{d}$ is the commanded lateral acceleration, whose 
upper bound  is expressed as
\begin{equation}
\begin{split}
   a_{\mathrm d} &\in \left[-\frac{V_u^2}{R_{\min}},\ \frac{V_u^2}{R_{\min}}\right]. 
\end{split}
\end{equation}
 The curvature $\kappa= 1/R_\mathrm{min}$ ($R_\mathrm{min} > 0$ is minimum turn radius) is related to the $a_\mathrm{d}$ and is saturated as
 \begin{equation}
\begin{split}
\kappa & =\frac{a_{\mathrm d}}{V_u^2}\in\left[-\frac{1}{R_{\min}},\ \frac{1}{R_{\min}}\right]. 
\end{split}
\label{eq:sat}
\end{equation}
Furthermore, path-following characteristics of the guidance method are evaluated using three performance indices, which are as follows:

\begin{align}
\begin{split}
 \text{Settling time }
t_{\mathrm s} = \inf\{\,t\ge 0:\ |d(t)|\le \varepsilon\,\} \ \text{in} \ \mathrm{s},\\
\text{Control effort }
J = \int_{0}^{t_f} a_{\mathrm d}^2(t)\,dt \ \text{in} \ [\mathrm{m}^2/\mathrm{s}^3],\\
\text{Peak overshoot }
Mp_{\varepsilon}
=
\max_{t\ge t_{\varepsilon}}\!\big(-\,d(t)-\varepsilon\big)^{+},
\\
x^{+}=\max\{x,0\}.   
\end{split}
\label{eq:metrics}
\end{align}
Here,  \emph{peak overshoot} is the amount by which the trajectory goes above/below the lower band $\varepsilon$ after the first entry, and $t_f$ is a prescribed evaluation horizon.

\section{Varying $L_0$ Guidance Logic} \label{sec:proposedguidance}

In this section, we briefly discuss the existing framework of $L_0/L_1$ guidance logic, followed by geometrical and theoretical discussion on the proposed variable $L_0$ guidance logic.  
\subsection{Background}
Given the vehicle's position $P (x,y)$ defined in an inertial frame, a typical schematic of $L_0$ guidance logic is depicted in Fig.~\ref{fig:L0 guidance logic}. Therein, we determine the closest point on the path to the vehicle as $O$, and the corresponding cross-track error as $d$. From that point, we advance a previously defined distance $L_0$ towards a virtual target $T$.
The vector joining the current position of the vehicle and the reference point $T$ is defined as the vector $L_1$ (LOS vector), and the angle between $L_1$ and the vehicle's velocity $V_u$ is defined as the heading error $\eta$, which is obtained as
\begin{equation}
\eta = \sin^{-1} \left( \frac{V_u \times L_1}{|V_u||L_1|}\right).
\label{eq: eta computed}
\end{equation}
\begin{figure}
\centering
\includegraphics[width=0.6\linewidth]{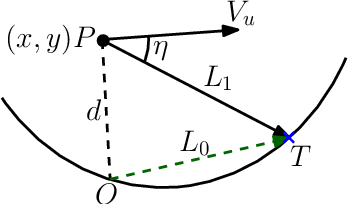}
\caption{Schematic of $L_0$ guidance logic.}
\label{fig:L0 guidance logic}
\end{figure} 

Following the guidance logic used in~\cite{fernandes_path-following_2024}, $L_1$ relates to the tuning parameter $L_0$ (constant) and the local path curvature $\kappa  $  at the closest point $O$ as
\begin{equation}
L_1 = \sqrt{d^2 + L_0^2\bigl(1+d\kappa\bigr)\,} \equiv L_1(d,\kappa),  \qquad L_1 < 2R.
\label{eq: L1 formula}
\end{equation}
For straight-line path, $\kappa=0$ and therefore Eq.\eqref{eq: L1 formula} simplifies to
\begin{align}
   L_1(d,\kappa) = \sqrt{d^2+L_0^2}
   \label{eq: straight line L0 guidance}
\end{align}
For further details, please see \cite{fernandes_path-following_2024}. While following the prescribed path, the centripetal acceleration required to follow an instantaneous circular arc formed by joining the UAV position and virtual target $T$ on the path is expressed as
\begin{equation} \label{GuidanceLogicEq}
a_{\mathrm{d}}=\frac{2 V_u^2 \sin(\eta)}{L_1(d,\kappa)} = \frac{2 V_u^2 \sin(\eta)}{\sqrt{d^2 + L_0(d)^2\bigl(1+d\kappa\bigr)\,}}
\end{equation}

\begin{figure}
    \centering
    \includegraphics[width=\linewidth]{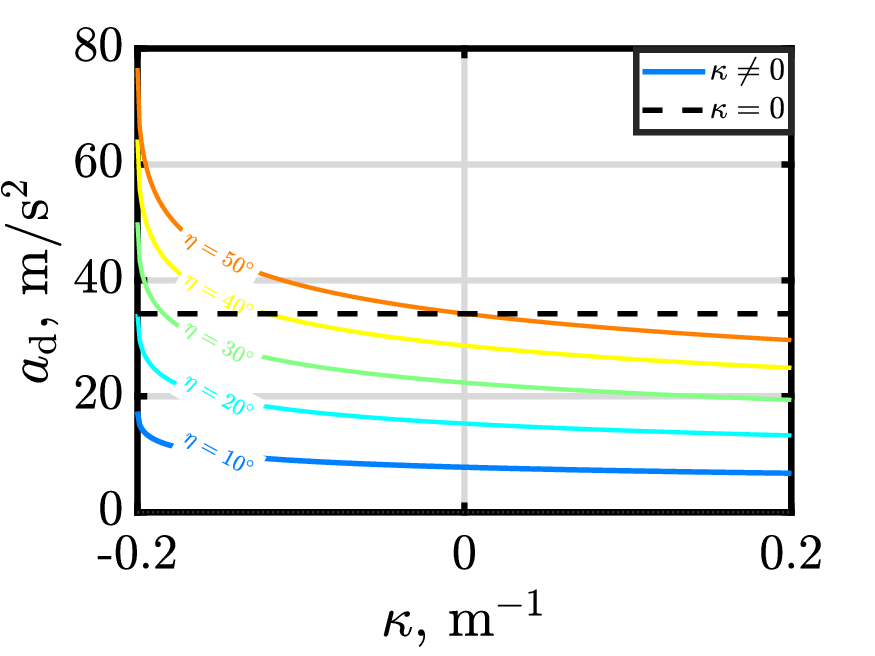}
    \caption{ Constant $L_0$ guidance with $(L_0, d) = (100, 50) $ m}
    \label{fig: constnat L0 versus L1 ad kappa}
\end{figure}
However, one must consider the vehicle's dynamic capability, which ultimately limits heading angle variations. In this regard, define the saturation boundary by constraining $\eta$ as
\begin{align}
    \bar\eta(d)=\arcsin\!\Big(\min\{1,\ L_1(d,\kappa)/(2R_\mathrm{min})\}\Big)\in(0,\pi/2)
    \label{eq: saturation boundary}
\end{align}
Set three regions as
\begin{equation}
\left\{
\begin{aligned}
& S_1 = \{ |\eta| \leq \bar{\eta}(d) \}                   &\qquad \text{(unsaturated)}   \\
& S_2 = \{ \eta > \bar{\eta}(d) \},                    \\
& S_3 = \{ \eta < -\bar{\eta}(d) \}                    &\qquad \text{(saturated)}
\end{aligned}
\right.
\label{eq: saturation and unsaturation region}
\end{equation}

Accordingly, the guidance command for the respective regions are expressed as
\begin{equation}
a_\mathrm{d} = 
\begin{dcases}
\dfrac{2V_u^{2}}{L_1(d,\kappa)} \sin \eta         & \text{if } (d,\eta) \in S_1, \\
\dfrac{2V_u^{2}}{L_1(d,\kappa)} \sin \bar{\eta}\sign(\bar{\eta})   & \text{if } (d,\eta) \in S_2, S_3\\
\end{dcases}
\label{eq: accleration for saturation and unsaturation region}
\end{equation}
In practice, estimating the virtual target on the path from onboard measurements and implementing the corresponding guidance commands inevitably introduces a projection error. A practical way to mitigate this is to constrain that error within a prescribed bound. To this end, we work in the local Frenet frame at the closest path point $O$, project the virtual target $T$ along the tangent at $O$ as $T'$, and explicitly quantify the resulting projection error $TT'$, which then guides the control design. On this geometric basis, we develop the variable $L_0$ guidance law, its construction and properties in the subsequent section.

\subsection{Feasibility set and geometric bounds}
\begin{figure}
\centering
\includegraphics[width=\linewidth]{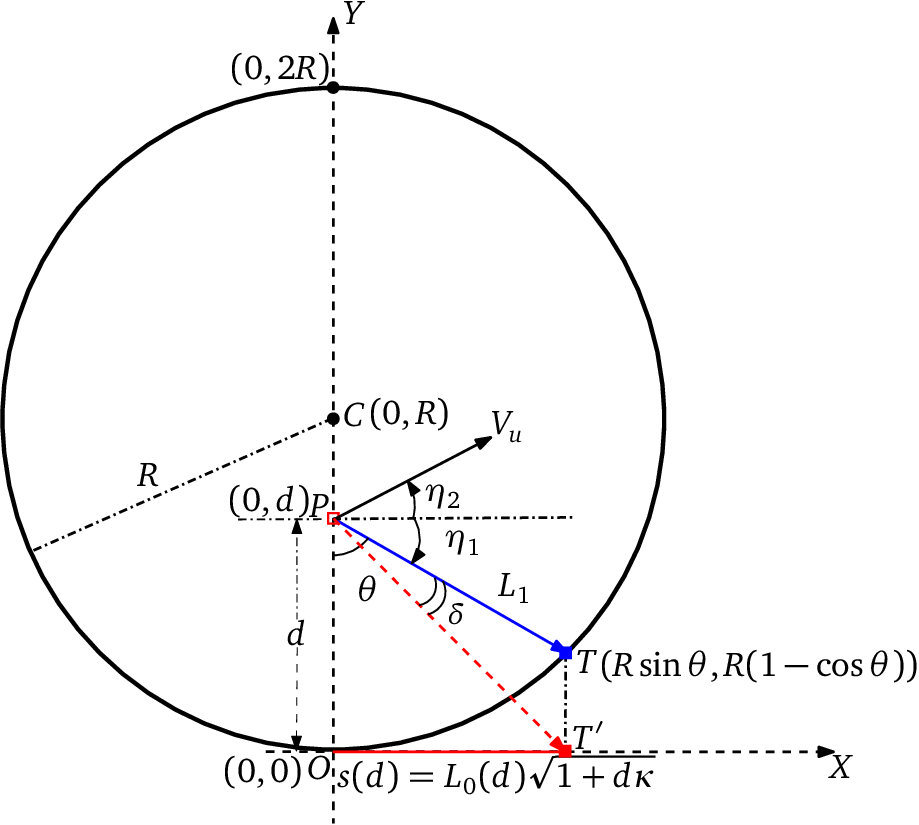}
\caption{$L_0$ feasibility study}
\label{fig:feasibility study}
\end{figure}
Consider $O\in\mathbb{R}^2$ be the closest point on the path, with tangent $OX$ and inward normal $OY$as shown in Fig.~\ref{fig:feasibility study}.
Let the vehicle be at $P=(0,d)$ and denote the local curvature at $O$ by $\kappa$ (radius $R=1/|\kappa|$).
We place a virtual target on the path ahead of $O$ by choosing a length $s$ along the tangent,
\begin{equation}
\label{eq:s-def}
s(d,\kappa)\ :=\ L_0(d)\,\sqrt{\,1+d\kappa\,}, \qquad 1+d\kappa>0,
\end{equation}
Thus, LOS vector is $L_1 = T'-P=(s,-d)$ and magnitude is computed as
\begin{equation}
\label{eq:L1-def}
L_1(d,\kappa)\ =\ \|T'-P\|\ =\ \sqrt{\,d^2 + L_0^2(d)\,(1+d\kappa)\,}.
\end{equation}
Accordingly, the projection error $\delta$ can be obtained as 
\begin{align}
\begin{split}
  \|\delta \|&=\|T-T'\|=\left\|\left(R\sin\theta-s\right),\ R\left(1-\cos\theta\right)\right\| \\
&\le \ \frac{|\kappa|\,s^2}{2}+O((\kappa s)^4)\   
\end{split}
\end{align}
To guarantee $\|\delta \|\le \varepsilon_{\mathrm{proj}}$, one may use
\begin{align}
s \ \le\  \sqrt{\tfrac{2\,\varepsilon_{\mathrm{proj}}}{|\kappa|}}\quad(\kappa\neq 0).    
\end{align}
A practical consideration for design is given as
\begin{align}
\begin{split}
  s &=\min\!\Big\{\,L_0\sqrt{1+d\kappa},\
\sqrt{\tfrac{2\,\varepsilon_{\mathrm{proj}}}{|\kappa|}}\Big\} 
\end{split}  
\end{align}
In this context, one can define the feasibility set while ensuring the virtual target lies ahead of the vehicle along the instantaneous circular arc.
\begin{definition}\label{def:feasible-set}
For any circular arc of radius $R$ and curvature $\kappa(s^\star)$ at the closest point $O$, define a feasibility set
\begin{align}
    \begin{split}
    \mathcal{F}\ :=\ \Big\{(d,\eta)\in\mathbb{R}\times(-\pi,\pi]:
\ 1+d\,\kappa(s^\star)>0, \\ \ L_1(d,\kappa)\le 2R-d \Big\}.    
    \end{split}
\end{align}
with $R=1/|\kappa(s^\star)|$ (and $R=\infty$ if $\kappa=0$), and
\begin{align}
\begin{split}
 L_1(d,\kappa)=\sqrt{\,d^2 + L_0^2\,\big(1+d\,\kappa(s^\star)\big)\,},
\\
s=L_0\sqrt{1+d\,\kappa(s^\star)}.   
\end{split}  
\end{align}

\end{definition}
Though the commanded lateral acceleration (unsaturated) is given by Eq.~\eqref{GuidanceLogicEq}, the local path curvature sensitivity with $a_\mathrm{d}$ is explained subsequently. 
\subsection{ Curvature sensitivity of $a_\mathrm{d}$}
For fixed $(V_u,\eta)$, using Eq.~\eqref{GuidanceLogicEq}
\begin{align}
\begin{split}
  \frac{\partial a_d}{\partial \kappa}\Big|_{\eta}
&= -\,\frac{2V_u^2\sin\eta}{L_1^2}\,\frac{\partial L_1}{\partial \kappa} 
= -\,\frac{2V_u^2\sin\eta}{L_1^2}\cdot \frac{L_0^2(d)\,d}{2L_1} \\
&= -\,\frac{V_u^2\,\sin\eta\,L_0^2(d)\,d}{\bigl(d^2+L_0^2(d)(1+d\kappa)\bigr)^{3/2}}.  
\end{split}
\end{align}
A dimensionless curvature sensitivity index (CSI) $\sigma$ is defined as
\begin{equation}
\sigma
=\frac{\kappa}{a_d}\frac{\partial a_d}{\partial \kappa}\Big|_{\eta}
= -\,\frac{\kappa\,L_0^2(d)\,d}{2\bigl(d^2+L_0^2(d)\,(1+d\kappa)\bigr).
\quad}
\label{eq: CSI sigma defined}
\end{equation}
Two useful limits are stated as follows:
\begin{enumerate}
   \item \textit{Small $d\kappa$ (first order, no $d/L_0$ assumption):}
\begin{align}
  \sigma \ \approx\ -\,\frac{\kappa\,L_0^2(d)\,d}{2\bigl(d^2+L_0^2(d)\bigr)}  
\end{align}
If in addition $|d|\ll L_0(d)$ (near-path), then
\begin{align}
    \sigma \ \approx\ -\,\frac{\kappa\,d}{2}\ 
\end{align}
\item \textit{Large $L_0$ (far field):}
\begin{align}
   \sigma \ \approx\ -\,\frac{\kappa d}{2(1+d\kappa)}\,,\qquad |\sigma|<\tfrac{1}{2} 
\end{align} 
\end{enumerate}
Thus Eq.~\eqref{GuidanceLogicEq} is curvature-sensitive, but the sensitivity is \emph{attenuated}
by increasing $L_0(d)$ far from the path.
For near-path linearization (small $|\eta|$), using $\sin\eta\simeq\eta = \frac{d}{L_1}$ and the standard small-angle $L_1$-geometry relation yields,
\begin{equation}
\label{eq:near-ad}
a_d \approx \frac{2V_u^2\,\eta}{L_1}
\approx \frac{2V_u^2\,d}{L_1^2}
= \;\frac{2V_u^2\,d}{d^2+L_0^2(d)\,(1+d\kappa)}\;.
\end{equation}
Thus using Eq.~\eqref{eq:near-ad},
\begin{equation}
\frac{\partial a_d}{\partial \kappa}
\approx -\,\frac{2V_u^2\,d^2\,L_0^2(d)}{\bigl(d^2+L_0^2(d)\,(1+d\kappa)\bigr)^{2}}
\
\end{equation}
which is always negative for $d\neq0$, and increasing $\kappa$ reduces $a_d$.
Fig.~\ref{fig: constnat L0 versus L1 ad kappa} illustrates that the variation in $a_\mathrm{d}$ is prominent with an increase in $|\kappa|$ while varying $\eta$.

\subsection{Guidance law}
In contrast to the constant $L_0$, we propose a varying $L_0$ profile  denoted as $\Lv(d)$, which is expressed as
\begin{align}
\Lv(d)&=\Lmin + (\Lmax-\Lmin)\!\left(1-e^{-|d|/d_c}\right) 
\label{eq: variable L0 guidance profile}
\end{align}
with parameters \(\Lmax>\Lmin>0\), \(d_c>0\). Table~\ref{tab:combined_effect} illustrate the effect of design parameter $(\Lmax,d_c)$ on the $L_0^{v}(d)$ and the resulting $a_\mathrm{d}$.
For a fair comparative analysis between the two guidance methods, we use constant $L_0$ denoted as $\Lc(d)$ as 
\begin{align}
 \Lc(d)&\equiv \Lmin  
\end{align}
The corresponding LOS lengths of the two guidance laws using Eq.~\eqref{eq: L1 formula} are expressed as
\begin{align}
L_1^{\mathrm{c}}(d,\kappa)&=\sqrt{d^2+{\Lmin}^2(1+d\kappa)}
\label{eq: L1 for constant L0},\\
L_1^{\mathrm{v}}(d,\kappa)&=\sqrt{d^2+{\Lv}^2(d)(1+d\kappa)}.
\label{eq: L1 for variable L0}
\end{align}
If $\Lv(d)\in[L_\mathrm{min},L_\mathrm{max}]$, with $L_\mathrm{min}>0$ and \(\Lv(d)\) is non decreasing in \(|d|\) then,
\begin{equation}
\label{eq:L1-minmax}
\sqrt{d^2+\Lmin^2(1+d\kappa)}\ \le\ L_1^{\mathrm{v}}(d,\kappa)\ \le\ \sqrt{d^2+\Lmax^2(1+d\kappa)}.
\end{equation}
\begin{lemma}[Monotone envelope]
\label{lem: monotone envelope}
For all feasible \((d,\kappa)\) with \(1+d\kappa>0\),
\begin{align}
L_1^{\mathrm{v}}(d,\kappa)\ \ge\ L_1^{\mathrm{c}}(d,\kappa).
\label{eq:  L1 variable greater than L1 constant}
\end{align}
\end{lemma}

\begin{proof}
Since mapping of $x\mapsto\sqrt{d^2+x^2(1+d\kappa)}$ is increasing when $1+d\kappa>0$ and $\Lv(d)\ge\Lmin$, hence leads to \eqref{eq:  L1 variable greater than L1 constant}.
\end{proof}
\begin{lemma}[Quantitative difference]
\label{lem:quant-gaps}
Subject to Lemma \ref{lem: monotone envelope} and LOS lengths $L_1^{\mathrm{c}}(d,\kappa)$, $L_1^{\mathrm{v}}(d,\kappa)$ given by Eq.\eqref{eq: L1 for constant L0} and \eqref{eq: L1 for variable L0}, respectively
then for each fixed $(d,\kappa)$
\begin{align}
\begin{split}
  \Delta L_1(d,\kappa)
& = L_1^{\mathrm{v}}(d,\kappa)-L_1^{\mathrm{c}}(d,\kappa) \\
&= \frac{{\Lv}^2(d)-{\Lmin}^2}{L_1^{\mathrm{v}}(d,\kappa) +L_1^{\mathrm{c}}(d,\kappa)}\,(1+d\kappa) \ \ge 0    
\end{split}
 \label{eq: difference in L1}
\end{align}
Thus, using Eq.~\eqref{eq: saturation boundary}, the difference in saturation region of heading error is deduced as
\begin{align}
    \Delta\bar\eta(d)
& = \bar\eta^{\mathrm v}(d)-\bar\eta^{\mathrm c}(d)
\ \ge\ \frac{\Delta L_1}{2\Rmin\sqrt{\,1-\bigl(L_1^{\mathrm v}/(2\Rmin)\bigr)^2\,}}.
\label{eq: heading error saturation derived}
\end{align}
Moreover, in the far field $|d|\gg d_c$ where $\Lv(d)\approx \Lmax$,
\begin{equation}
\Delta\bar\eta \;\approx\; \arcsin\!\frac{\Lmax}{2\Rmin} \;-\; \arcsin\!\frac{\Lmin}{2\Rmin},
 \label{eq: heading error farfield}
\end{equation}
and the time difference is 
\begin{align}
\Delta T_{\mathrm{far}} \;\approx\; -\,\frac{\Rmin}{V_u}\,\Delta\bar\eta.
\label{eq: time difference farfield}
\end{align}
\end{lemma}
\begin{proof}
  From the fact that
$(L_1^{\mathrm v}(d,\kappa))^2-{L_1^{\mathrm c}(d,\kappa)}^2 = \big({\Lv}^2(d,\kappa)-{\Lmin}^2\big)(1+d\kappa)$ with the factorization reduces to Eq.~\eqref{eq: difference in L1} as
$L_1^{\mathrm v}(d,\kappa)-L_1^{\mathrm c}(d,\kappa)=\frac{(L_1^{\mathrm v}(d,\kappa))^2-(L_1^{\mathrm c}(d,\kappa))^2}{L_1^{\mathrm v}(d,\kappa)+L_1^{\mathrm c}(d,\kappa)}$.
Next, using $\bar\eta=\arcsin(L_1/(2\Rmin))$ on the unsaturated branch and hence,
$\frac{d}{dL_1}\arcsin(L_1/(2\Rmin))=\frac{1}{2\Rmin\sqrt{1-(L_1/(2\Rmin))^2}}$. Consequently, evaluating at
$L_1(d,\kappa) \\ = L_1^{\mathrm v}(d,\kappa)$ and applying the mean-value inequality yields the bound \eqref{eq: heading error saturation derived}.
Finally, when $|d|\gg d_c$ we have $\Lv(d,\kappa)\approx\Lmax$, giving \eqref{eq: heading error farfield} and the time
reduction in saturation from $|\dot{\eta}|\ge V_u/\Rmin$.
\end{proof}
The theoretical analysis discussed above is simulated with parameters $(V_u, g, L_{\min}, L_{\max}, d_c)
= (50~\text{m/s}, 9.8~\text{m/s}^2, \\ 50~\text{m}, 150~\text{m}, 30~\text{m}) $
for \( R_{\min} = 100~\text{m} \), \(\eta \in (-\pi,\, \pi]~\text{rad}\), and \( d \in [0,\,200]~\text{m} \). The saturation boundary $\bar{\eta}$  follow Eq.~\eqref{eq: saturation boundary}, wherein, $L_1(d,\kappa)$ will be substituted by Eq.~\eqref{eq: L1 for constant L0} and \eqref{eq: L1 for variable L0} for the two guidance laws. The corresponding saturation regions satisfy Eq.~\eqref{eq: saturation and unsaturation region}. 
To quantify the improvement contributed by the variable $L_0(d)$ profile over the constant $L_0$ guidance law, we define two complementary metrics based on the unsaturated region area in the $(d, \eta)$ plane.
Let $\mathcal{A}_{\text{const}}$ denote the fraction of the $(d, \eta)$ plane where the constant $L_0$ guidance law remains unsaturated, and $\mathcal{A}_{\text{var}}$ denote the corresponding fraction for the variable $L_0(d)$ profile. These are computed as:
\begin{align}
\mathcal{A}_{\text{const}} &= \frac{1}{N} \sum_{i=1}^{N} \mathbb{1}\left( |\eta_i| < \bar{\eta}_{\text{const}}(d_i) \right), 
\label{eq: area for constant L0 guidance law}\\
    \mathcal{A}_{\text{var}} &= \frac{1}{N} \sum_{i=1}^{N} \mathbb{1}\left( |\eta_i| < \bar{\eta}_{\text{var}}(d_i) \right)
    \label{eq: area for variable L0 guidance law}
\end{align}
where $N$ is the total number of grid points, $\mathbb{1}(\cdot)$ is the indicator function, and $\bar{\eta}(d)$ represents the saturation boundary in heading error.

\textit{Absolute Gain (percentage points):} This metric measures the direct increase in the unsaturated region as a difference in percentage points described by
\begin{equation}
    G_{\text{abs}} = \left( \mathcal{A}_{\text{var}} - \mathcal{A}_{\text{const}} \right) \times 100\%.
\end{equation}
The absolute gain provides an intuitive measure of how many additional percentage points of the $(d, \eta)$ plane remain unsaturated when using the variable $L_0(d)$ strategy.

\textit{Relative Gain (percentage increase):} This metric quantifies the proportional improvement relative to the baseline derived as
\begin{equation}
    G_{\text{rel}} = \left( \frac{\mathcal{A}_{\text{var}}}{\mathcal{A}_{\text{const}}} - 1 \right) \times 100\%.
\end{equation}
The relative gain expresses how much larger the unsaturated region becomes as a percentage of the original constant $L_0$ baseline, providing a normalized measure of effectiveness. Fig.~\ref{fig: Comparative study for ratio equal to 3} illustrates the relative comparison of operational envelope. Fig.~\ref{fig: constant L0 saturation and unsaturation region} and \ref{fig: variable L0 saturation and unsaturation region} correspond to operational regions for constant $L_0$ and variable $L_0$ guidance laws, respectively. At the ratio $L_{\max}/L_{\min} = 3.0$ (corresponding to $L_{\max} = 150$~m for $L_{\min} = 50$~m), the variable lookahead strategy achieves $\mathcal{A}_{\text{abs}} = 17.32\%$  and $\mathcal{A}_{\text{rel}} = 72.58\%$ , demonstrating that the unsaturated region expands from $23.86\%$ to $41.17\%$ of the $(d,\eta)$ plane as gain region indicated in green in Fig.~\ref{fig: advantage of variable L0 over constant L0 for ratio 3}. Furthermore, a detailed analysis of the ratio $ \in [1,5]$ is shown in Fig.~\ref{fig: Percentage increase in unsaturation region for different ratios}. The absolute gain (green, left axis) shows the direct benefit in percentage points, while the relative gain (magenta, right axis) demonstrates the proportional enhancement. Both metrics increase with the ratio, with marginal benefits observed at higher values, suggesting an optimal design range for practical implementation.

\begin{corollary}
\label{cor:wedge}
With \(\bar\eta^{\mathrm{c}}(d),\bar\eta^{\mathrm{v}}(d)\) computed from \(L_1^{\mathrm{c}}(d,\kappa), \\ L_1^{\mathrm{v}}(d,\kappa)\), from Eq.\eqref{eq: L1 for constant L0} and \eqref{eq: L1 for variable L0},
\begin{align}
\bar\eta^{\mathrm{v}}(d)\ \ge\ \bar\eta^{\mathrm{c}}(d)\quad\Rightarrow\quad
S_1^{\mathrm{v}}\ \supseteq\ S_1^{\mathrm{c}}.    
\end{align}
Figure~\ref{eq: polar view of eta d space} illustrates the polar-like representation of comparative study, where each point in the circular plot corresponds to a state $(d,\eta)$ represented in Cartesian coordinates as $(d\cos\eta, d\sin\eta)$. The concentric dashed circles indicate constant cross-track error levels ($d = 50$~m and $d = 100$~m), while the outer solid circle marks the boundary at $d = 150$~m.

Figure~\ref{fig: polar view constant L0} illustrates that the unsaturated region $S_1^c$ (light green) forms a relatively narrow band around the horizontal axis ($\eta \approx 0$). The saturation boundary $\bar{\eta}(d)$ (solid black curves) defines the transition between the unsaturated region and the saturated regions $S_2^c$ (light red, $\eta > \bar{\eta}$) and $S_3^c$ (light blue, $\eta < -\bar{\eta}$).
In Figure~\ref{fig: polar view variable L0}, the unsaturated region $S_1^v$ (light green) expands significantly compared to Figure~\ref{fig: polar view constant L0}, as discussed previously and analyzed in detail in Figure~\ref{fig: Comparative study for ratio equal to 3}. 
The circular representation clearly illustrates the geometric advantage of the variable $L_0$ strategy: the unsaturated region $S_1^v$ strictly contains $S_1^c$ (i.e., $S_1^v \supseteq S_1^c$), confirming Corollary~2. The added region
$S_1^v \setminus S_1^c$ represents states for which the variable $L_0$
 formulation remains unsaturated even though the constant $L_0$ baseline would saturate.

\end{corollary}

\begin{lemma}[Saturated exit time bound]
\label{lem:Tfar}
From any initial \((d(0),\eta(0))\),
\begin{align}
 T_{\mathrm{far}}(L_0)\ \le\ \frac{\Rmin}{V_u}\,\bigl(|\eta(0)|-\bar\eta(d(0))\bigr)^+ .   
 \label{eq: T far}
\end{align}

Hence \(T_{\mathrm{far}}^{\mathrm{v_u}}\le T_{\mathrm{far}}^{\mathrm{c}}\). Further,
\[
\int_0^{T_{\mathrm{far}}} a^2\,dt \;\le\; \bigl(V_u^2/\Rmin\bigr)^2\,T_{\mathrm{far}},
\]
so the saturated-phase control effort is no greater for the variable profile.
\end{lemma}

\begin{proof}
In saturation \(|a|=V^2/\Rmin\) and \(|\dot\eta|\ge V/\Rmin\). The angle distance to the boundary is reduced when \(\bar\eta(d)\) is larger, which holds by Corollary~\ref{cor:wedge}.
\end{proof}

\begin{figure}
\begin{subfigure}[b]{.25\textwidth}
	\centering			
   \includegraphics[width=\linewidth]{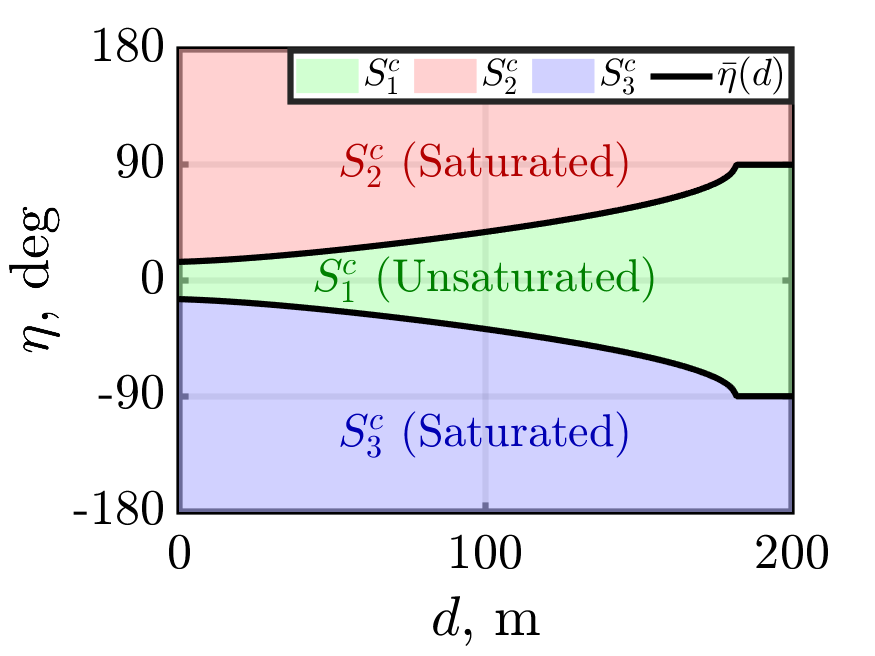}
	\caption{Constant $L_0$ guidance }         \label{fig: constant L0 saturation and unsaturation region}
\end{subfigure}%
\begin{subfigure}[b]{.25\textwidth}
	\centering
	\includegraphics[width=\linewidth]{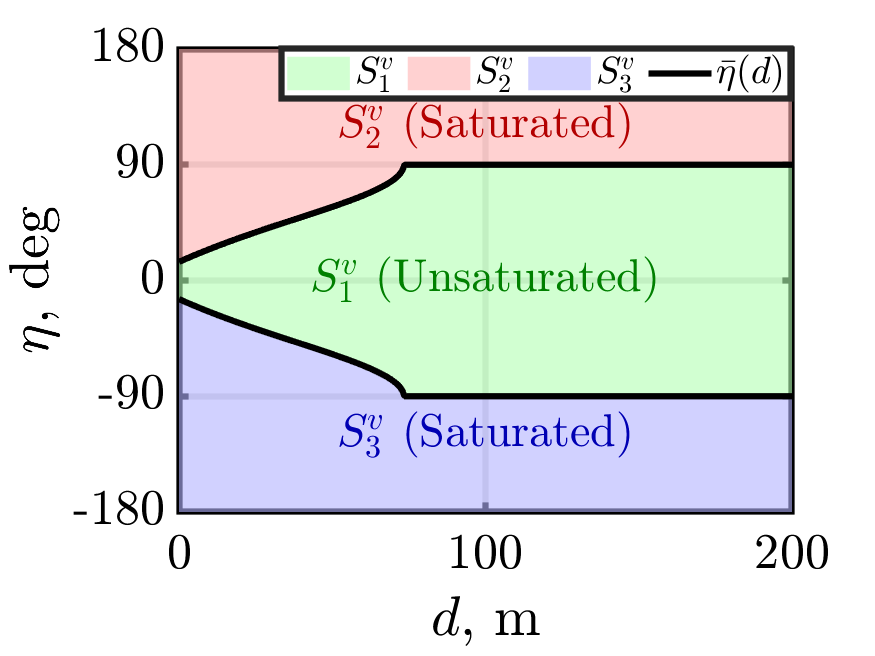}
	\caption{ Variable $L_0$ guidance }			\label{fig: variable L0 saturation and unsaturation region}       
\end{subfigure}%
\qquad 
\begin{subfigure}[b]{.25\textwidth}
	\centering	\includegraphics[width=\linewidth]{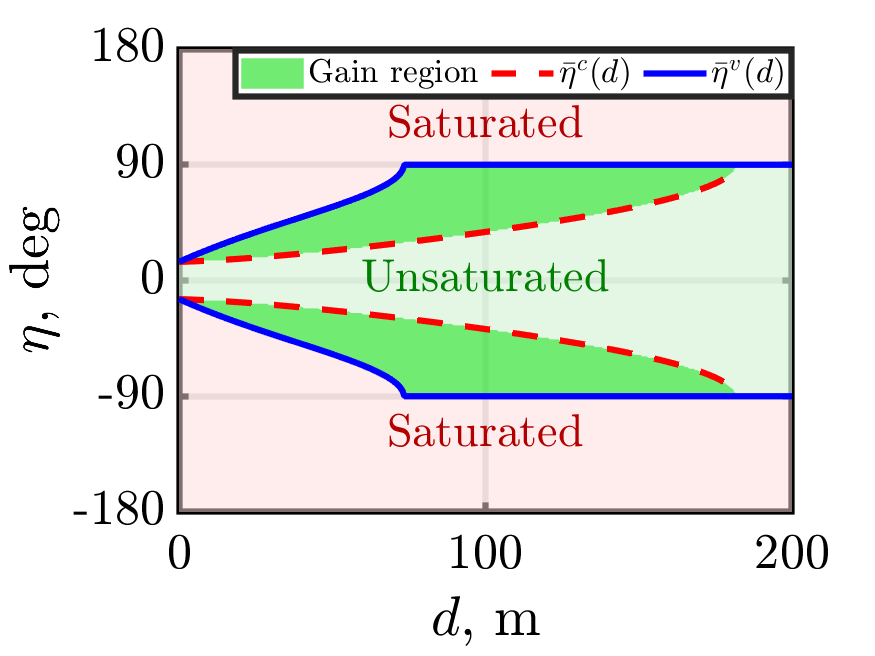}
	\caption{ Advantage for $\frac{L_{\mathrm{max}}}{L_{\mathrm{min}}} = 3 $.}			\label{fig: advantage of variable L0 over constant L0 for ratio 3}       
\end{subfigure}%
\begin{subfigure}[b]{.25\textwidth}
	\centering    \includegraphics[width=\linewidth]{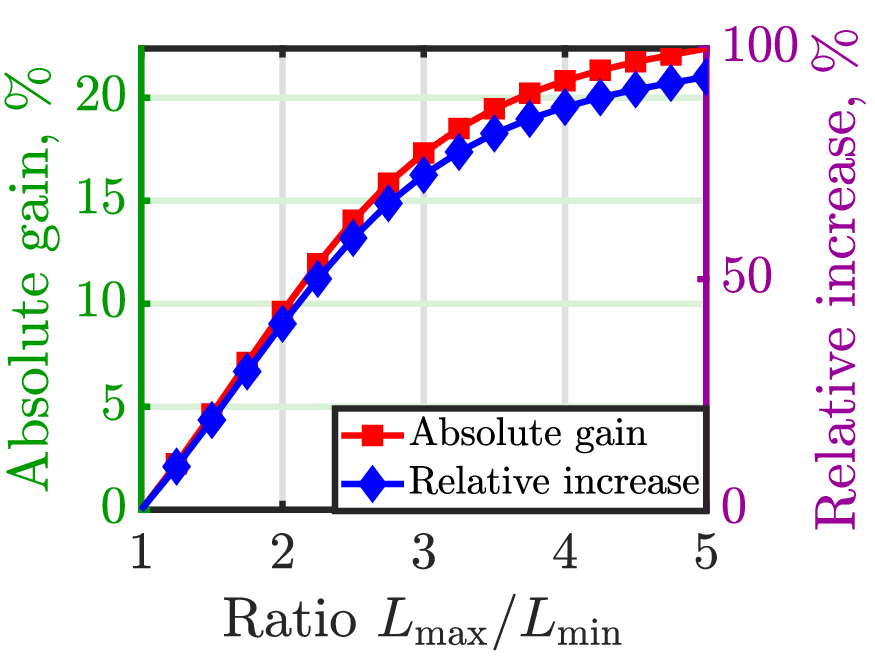}
    \caption{Advantage  for different ratios}
    \label{fig: Percentage increase in unsaturation region for different ratios}
\end{subfigure}%
\caption{Enlarged envelope analysis of two guidance laws}
\label{fig: Comparative study for ratio equal to 3}
\end{figure}

\begin{figure}
\begin{subfigure}[b]{.25\textwidth}
	\centering			
   \includegraphics[width=1.15\linewidth]{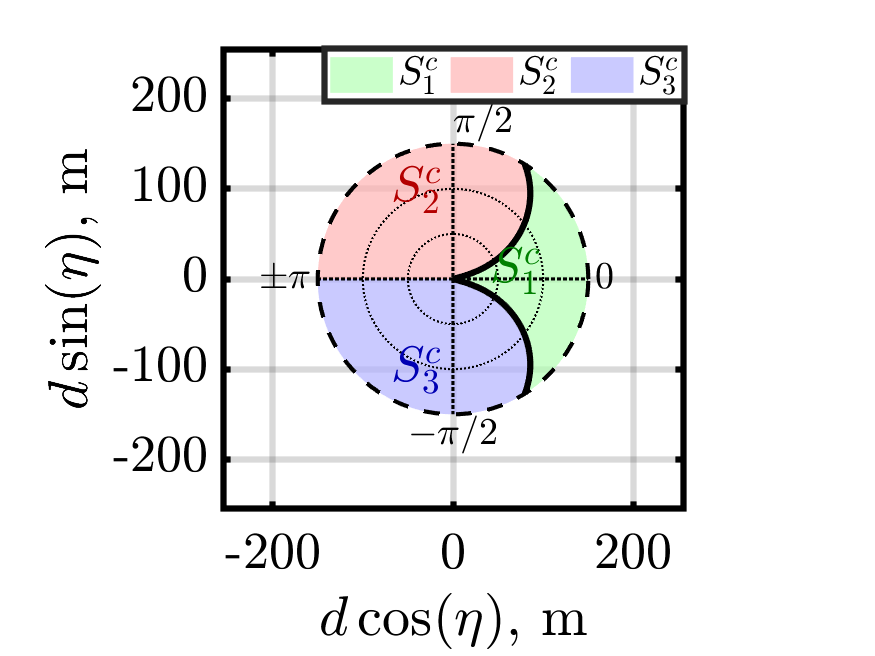}
	\caption{Constant $L_0$ guidance }         \label{fig: polar view constant L0}
\end{subfigure}%
\begin{subfigure}[b]{.25\textwidth}
	\centering
	\includegraphics[width=1.15\linewidth]{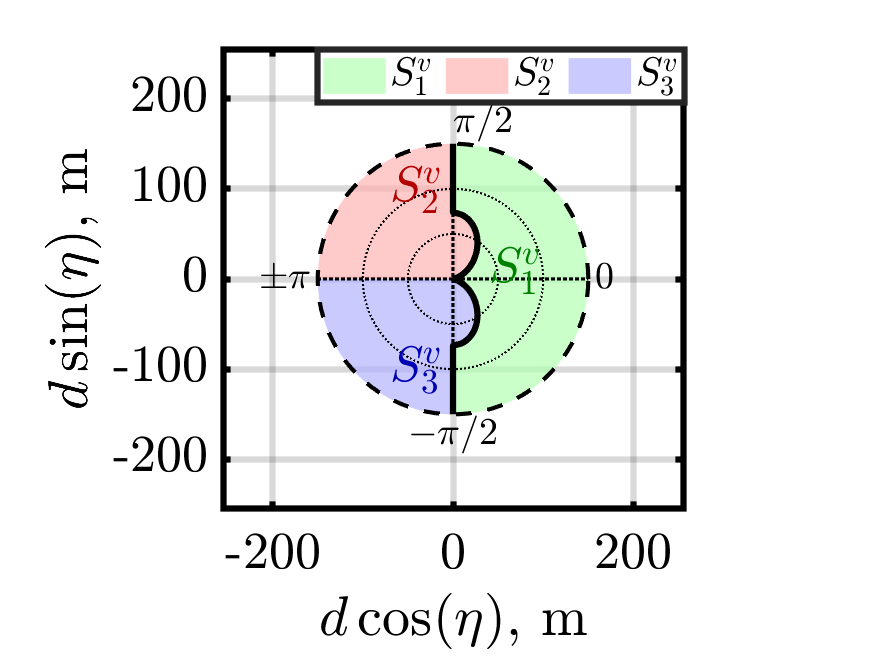}
	\caption{ Variable $L_0$ guidance }			\label{fig: polar view variable L0}       
\end{subfigure}%
\caption{Polar like representation of $\eta-d$ space for $\frac{L_{\mathrm{max}}}{L_{\mathrm{min}}} = 3 $. }
\label{eq: polar view of eta d space}
\end{figure}

\begin{lemma}[Near–path linear rate]
\label{lem:tau}
In $S_1$ both laws share the same local time constant
\begin{equation}
\label{eq:tau}
\tau\ =\ \frac{L_{\min}}{V},
\end{equation}
since $L_0^{\rm v}(0)=L_{\min}$ and $L_1\to L_{\min}$ near the path.
\end{lemma}

\begin{table*}
\centering
\caption{Combined effect of $L_\mathrm{max}$ and $d_c$ on variable $L_0$ guidance law performance}
\label{tab:combined_effect}
\begin{tabular}{ccccc}
	\toprule
	\textbf{Parameter} & \textbf{Effect on $L_0$} & \textbf{Effect on $a_d$} & \textbf{Characteristics} & \textbf{Remarks} \\
	\midrule
	Higher $ L_\mathrm{max}$ & Larger $L_0$ at large $d$ &  Lower acceleration & Slower, smoother &  Stable flight \\
	Lower $ L_\mathrm{max}$ & Smaller $L_0$ at large $d$ &  Higher acceleration & Faster, aggressive & Performance priority \\
	Lower $ d_c$ & Faster saturation &  Lower at medium $d$ & Slower from medium $d$ & Minimize oscillations \\
	Higher $ d_c$ & Slower saturation &  Higher at medium $d$ & Faster from medium $d$ & Fast initial convergence \\
	\bottomrule
\end{tabular}
\end{table*}

\subsection{Stability properties}

\begin{lemma}
\label{lem:g_def}
Let
\begin{align}
\begin{split}
L_1^2(d,\kappa)& =d^2+L_0(d)^2\,(1+d\kappa)     
\end{split}  
\end{align}
with $L_0:\mathbb R\to\mathbb R_+$ of class $C^1$. Then
\begin{align}
 \frac{\partial}{\partial d}\bigl(L_1^2(d,\kappa)\bigr)
=2\Bigl[d+(1+d\kappa)L_0(d)L_0'(d)+\tfrac12\,\kappa\,L_0(d)^2\Bigr]   
\end{align}
In particular, defining
\begin{align}
g(d)= d+(1+d\kappa)L_0(d)L_0'(d)+\tfrac12\,\kappa\,L_0(d)^2.    
\end{align}
we have $\frac{\partial}{\partial d}(L_1^2)=2g(d)$
\end{lemma}

\begin{proof}
Differentiate $L_1^2(d,\kappa)=d^2+L_0(d)^2(1+d\kappa)$ using the product rule leads to
\begin{align}
\begin{split}
  \frac{\partial}{\partial d}\bigl(L_1^2\bigr)
&=2d+2L_0 L_0'(1+d\kappa) + L_0^2\kappa \\
&=2\Bigl[d+(1+d\kappa)L_0L_0'+\tfrac12 \kappa L_0^2\Bigr]    
\end{split}
\end{align}
\end{proof}

\begin{lemma}
\label{lem:g_sign}
Suppose $L_0(\cdot)$ is even and nondecreasing in $|d|$, and the feasibility condition $1+d\kappa>0$ holds (e.g., on the standard feasible set).
Then $g(d)$ has the same sign as $d$, and for
\begin{align}
\Phi(d):=\int_{0}^{d}\frac{g(\xi)}{L_1(\xi,\kappa)^2}\,d\xi    
\end{align}
we have $\Phi(d)\ge 0$ with equality iff $d=0$.
\end{lemma}

\begin{proof}
By Lemma~\ref{lem:g_def}, $2g(d)=\partial_d L_1^2(d,\kappa)$. With $L_0$ even and nondecreasing in $|d|$ and $1+d\kappa>0$, the map $d\mapsto L_1^2(d,\kappa)$ is strictly increasing for $d>0$ and strictly decreasing for $d<0$, hence $g(d)$ has the sign of $d$. Since $L_1^2>0$, the integrand $\frac{g(\xi)}{L_1(\xi,\kappa)^2}$ has the sign of $\xi$, so the integral from $0$ to $d$ is nonnegative and only vanishes at $d=0$.
\end{proof}

\begin{lemma}
\label{lem:key_identity}
With $g(d)$ as in Lemma~\ref{lem:g_def} and the geometry above, the following identity holds:
\begin{align}
 \frac{g(d)}{L_1(d,\kappa)}
=\sin\eta_1\,\cos\eta_2+\cos\eta_1\,\sin\eta_2-\sin\eta_2\,\cos(\eta_1+\eta_2).   
\end{align}
\end{lemma}

\begin{proof}
Use the identities $\sin(\eta_1+\eta_2)=\sin\eta_1\cos\eta_2+\cos\eta_1\sin\eta_2$ and $\cos(\eta_1+\eta_2)=\cos\eta_1\cos\eta_2-\sin\eta_1\sin\eta_2$. A direct substitution of $\sin\eta_1=d/L_1$ and $\cos\eta_1=L_0\sqrt{1+d\kappa}/L_1$ into the right-hand side followed by collecting terms yields
\begin{align}
\begin{split} \sin\eta_1\cos\eta_2+\cos\eta_1\sin\eta_2-\sin\eta_2\cos(\eta_1+\eta_2) \\
=\frac{1}{L_1}\Bigl[d+(1+d\kappa)L_0L_0'+\tfrac12 \kappa L_0^2\Bigr]   
\end{split}   
\end{align}
where the bracket equals $g(d)$ by Lemma~\ref{lem:g_def}. 
\end{proof}
\begin{theorem}
Consider a vehicle moving with a constant speed $V_u>0$ near a regular $C^2$ path and the reduced Frenet error dynamics described by
\begin{align}
   \begin{split}
    \dot d=-V_u\sin\eta_2,\qquad
\dot\eta_2=\frac{a_\mathrm{d}}{V_u}, \\
\dot\eta_1=-\frac{V_u}{L_1}\sin(\eta_2-\eta_1),\\
\dot\eta=\frac{a_\mathrm{d}}{V_u}+\frac{V_u}{L_1}\sin\eta.    
   \end{split}    
\label{eq: error dynamics}
\end{align}
For every initial $(d(0),\eta(0))\in\mathcal F$, the solution is forward complete using heading guidance Eq.~\eqref{eq: accleration for saturation and unsaturation region}, and hence
\begin{align}
  (d(t),\eta(t))\to (0,0)\qquad \text{as }t\to\infty .
  \label{eq: solution dynamics}
\end{align} 


\end{theorem}

\begin{proof}
We present three cases that correspond to the regions $S_2\cup S_3$, the boundary $|\eta|=\bar\eta(d)$, and the unsaturated region $S_1$.

\smallskip
\noindent\emph{Case 1 (Saturated $\to$ Unsaturated in finite time):}
While $(d,\eta)\in S_2\cup S_3$, the applied acceleration has constant magnitude $|a|=V_u^2/R_{\min}$.
Hence, from \eqref{eq: error dynamics},
\begin{align}
\begin{split}
 \left|\dot\eta\right|
&=\left|\frac{a_\mathrm{d}}{V_u}+\frac{V_u}{L_1}\sin\eta\right|
\;\ge\; \frac{V_u}{R_{\min}}-\frac{V_u}{L_1}\sin\bar\eta \\
&=\frac{V_u}{R_{\min}}-\frac{V_u}{L_1}\frac{L_1}{2R_{\min}}
=\frac{V_u}{2R_{\min}}.    
\end{split} 
\label{eq: step A}
\end{align}
Therefore the angular distance from $|\eta(0)|$ to $\bar\eta\big(d(0)\big)$ is covered in finite time given by Eq.~\eqref{eq: T far}.
Thus every trajectory starting in $S_2\cup S_3$ enters $S_1$ in finite time.

\smallskip
\noindent\emph{Case 2 (Invariance of $S_1$):}
On $|\eta|=\bar\eta(d)$, using the saturated law,
\begin{align}
\dot\eta=\frac{2V_u}{L_1}\sin\bar\eta\,\mathrm{sgn}(\eta)+\frac{V_u}{L_1}\sin\eta
=\frac{V_u}{L_1}\big(\sin\bar\eta\,\mathrm{sgn}(\eta)+\sin\eta\big)   \end{align}

which is zero at $\eta=\pm\bar\eta$ and points toward the interior for small excursions beyond the boundary. Hence $S_1$ is positively invariant.

\smallskip
\noindent\emph{Case 3 (Lyapunov analysis in $S_1$):}
Inside $S_1$ the guidance law satisfy $a_\mathrm{d}=\tfrac{2V_u^2}{L_1}\sin\eta$. Define the Lyapunov function as
\begin{equation}
\mathcal V(d,\eta)=\tfrac12 V_u^2\sin^2\eta\ +\ V_u^2\Phi(d),
\label{eq:V-def}
\end{equation}
which is positive definite on $\mathcal F$ with unique minimum at $(d,\eta)=(0,0)$.
Differentiate $\mathcal V$ in Eq.\eqref{eq:V-def} along Eq.\eqref{eq: error dynamics} and on simplifying using Lemma~\ref{lem:g_def}-\ref{lem:key_identity} results in
\begin{align}
\dot{\mathcal V}=-\frac{2V_u^3}{L_1(d,\kappa)}\sin^2\eta_2\cos\eta_1 \le 0.    
\end{align}
In the feasible set, $s\ge 0$ so $\cos\eta_1=s/L_1(d,\kappa) \ge 0$, equality holds only at $\sin\eta=0$. Hence, $\mathcal V$ is nonincreasing in $S_1$.
By LaSalle’s invariance principle, the $\omega$–limit set is contained in $\{\sin\eta=0\}\cap S_1$, i.e., $\eta=0$.
On $\eta=0$ we have $\dot d=0$ and $\dot{\mathcal V}=0$, which implies $g(d)=0$. Since $g(d)$ has the sign of $d$,
this implies $d=0$. Thus, the only invariant set with $\dot{\mathcal V}=0$ is $(d,\eta)=(0,0)$, and hence the
equilibrium is asymptotically stable. 
\smallskip
Combining Steps A–C yields forward completeness and $(d(t),\eta(t))\to(0,0)$ for all initial
conditions in $\mathcal F$.
\end{proof}

\section{Simulation Results} \label{sec:simulations}

This section discusses numerical simulation studies for constant $L_0$ and varying $L_0$ guidance methods, while following straight-line and elliptic paths.
 In doing so, three performance indices, such as settling time, the control effort, and peak overshoot, are analyzed with variation in tunable parameters $(L_\mathrm{min}, L_{\mathrm{max}},d_c)$ over a total run time $t_{\mathrm{f}}$. Unless stated, the speed of the UAV is considered as $V_u = 12$ m/s. Simulation parameters for constant  $L_0 =L_\mathrm{min}$ and  the variable $L_0$ as governed by Eq.~\eqref{eq: variable L0 guidance profile},  are listed in Table~\ref{tab:guidance_params}. 

\begin{table}[ht]
  \captionsetup{width=\textwidth}
  \caption{Simulation parameters }
  \label{tab:guidance_params}
  \begin{tabular}{@{}lccc@{}}
    \toprule
    \textbf{Case}
      & \textbf{Initial conditions}
      & \textbf{Constant\ $L_0$}
      & \textbf{Variable\ $L_0$} \\
    \textbf{study}&  $(x_0,y_0)$ m, $\psi_0$ deg
      & $L_0 = L_{\min} $ m
      & $(L_{\max}$, $d_c)$ m \\
    \midrule
    Straight line
      & $(-150,\,50),\;90$
      & $40$
      & $82,\;32$ \\
    Elliptic path
      & $(250,\,120),\;150$
      & $22$
      & $100,\;20$ \\
    \bottomrule
  \end{tabular}
\end{table}

\begin{figure}[htbp]
\begin{subfigure}[b]{.25\textwidth}
	\centering			\includegraphics[width=\linewidth]{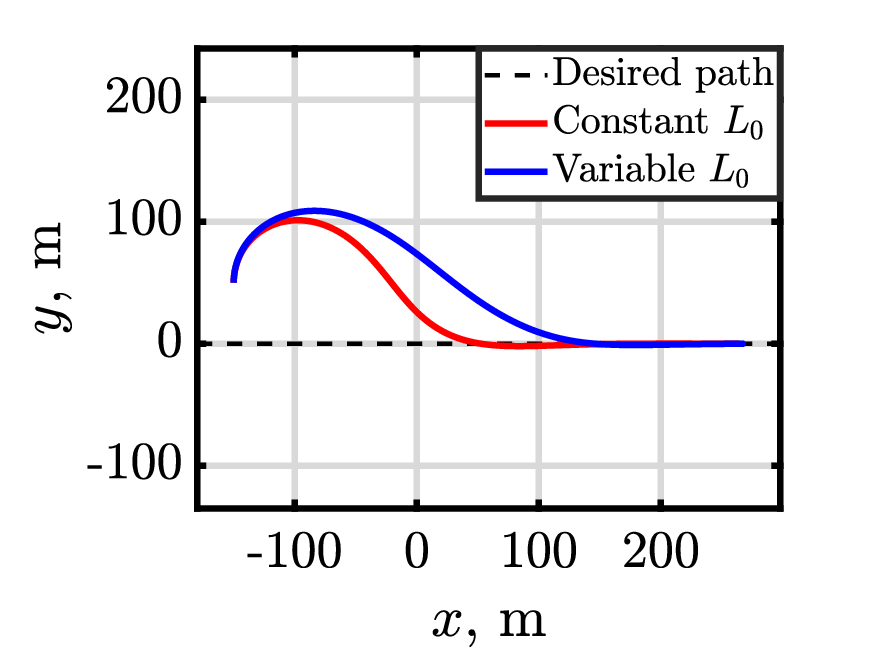}
	\caption{ UAV trajectory}		\label{fig:  UAV trajectory line}       
\end{subfigure}%
\begin{subfigure}[b]{.25\textwidth}
	\centering			\includegraphics[width=\linewidth]{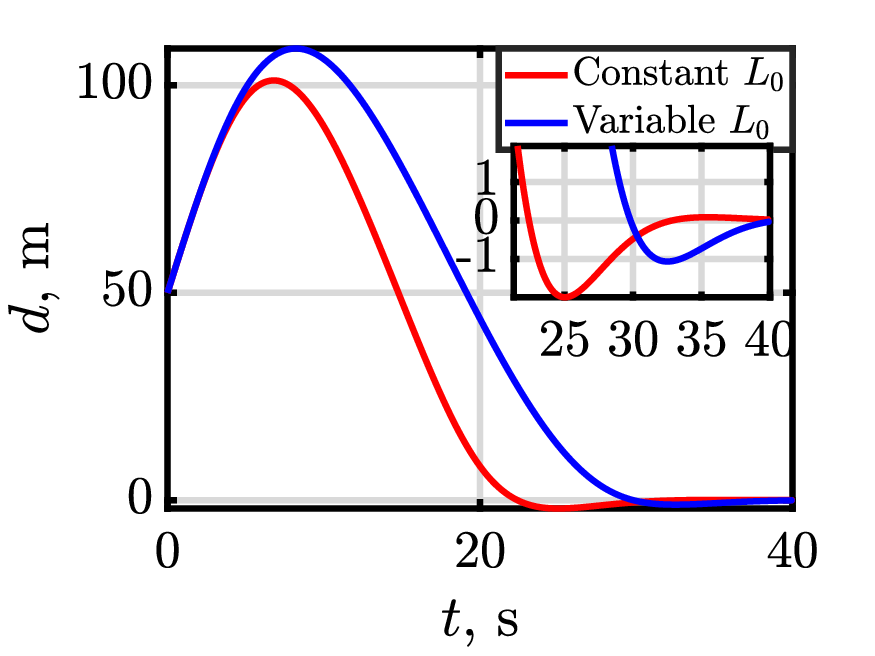}
	\caption{Tracking error}         \label{fig: Tracking error line}
\end{subfigure}%
\qquad
\begin{subfigure}[b]{.25\textwidth}
	\centering			\includegraphics[width=\linewidth]{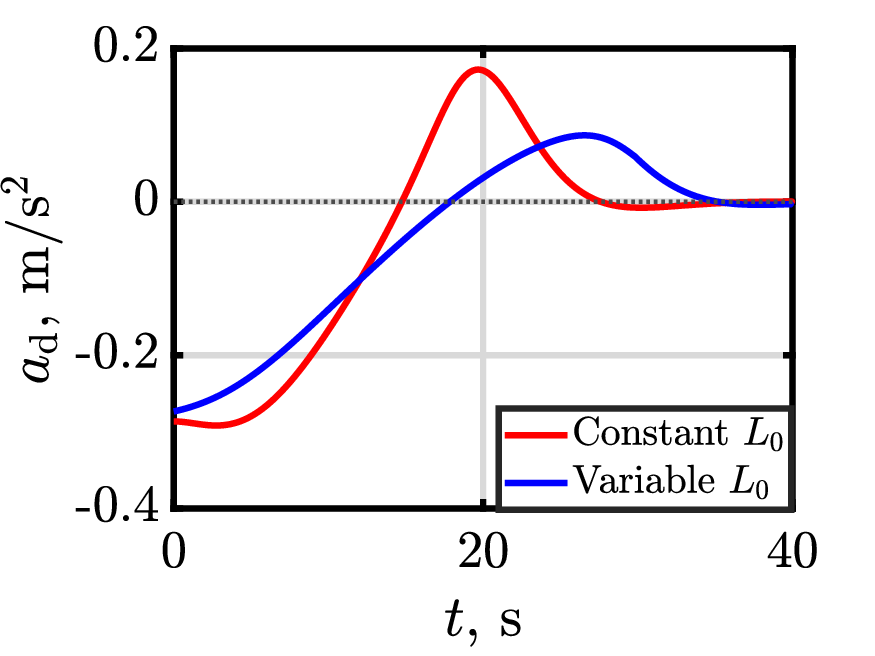}
	\caption{Lateral acceleration}         \label{fig: Lateral acceleration line}
\end{subfigure}%
\begin{subfigure}[b]{.25\textwidth}
	\centering
	\includegraphics[width=\linewidth]{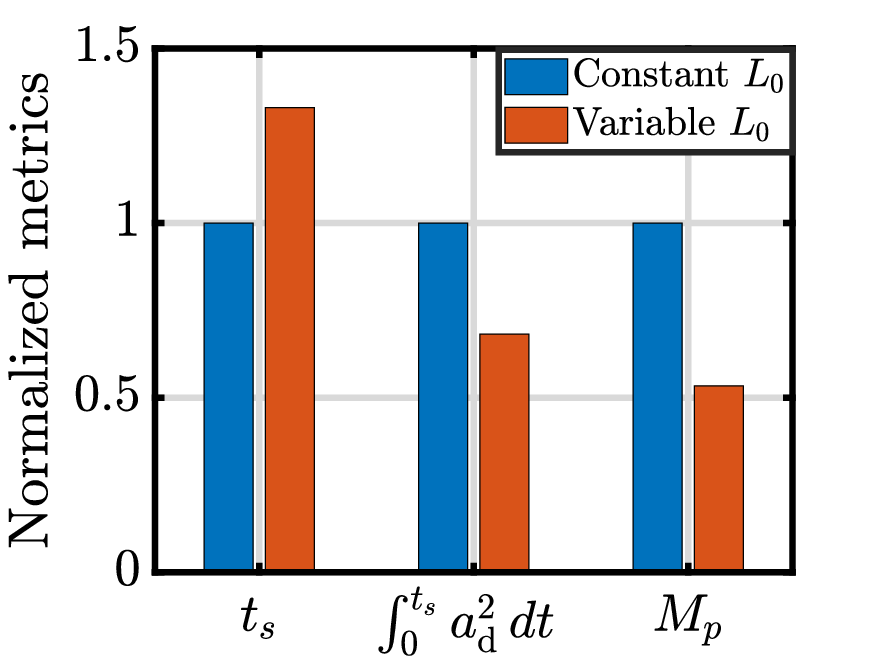}
	\caption{ Performance comparison}			\label{fig: Performance comparison line}       
\end{subfigure}%
\caption{Results for straight line following}
\label{fig: Results for straight line following}
\end{figure}
In the case study of straight line path following, the desired path is described by $y=0$. The numerical simulation results are shown in Fig.~\ref{fig: Results for straight line following}. Fig.~\ref{fig:  UAV trajectory line} plots the UAV trajectory, and the corresponding tracking error is plotted in Fig.~\ref{fig: Tracking error line}. While both methods reduces the tracking error  to zero and eventually align with the desired line, the variable $L_0$ approach exhibits significantly smaller overshoot in comparison to constant $L_0$ as indicated in Fig.~\ref{fig: Tracking error line}. The commanded lateral acceleration is depicted in Fig.~\ref{fig: Lateral acceleration line}.  It is noteworthy to highlight that, the lateral acceleration profile demonstrates a substantial reduction in control effort, as reflected in both the lower peak acceleration as in Fig.~\ref{fig: Lateral acceleration line} and the reduced integrated squared acceleration metric indicated in Fig.~\ref{fig: Performance comparison line}. 

\begin{figure}[htbp]
\begin{subfigure}[b]{.25\textwidth}
	\centering			\includegraphics[width=\linewidth]{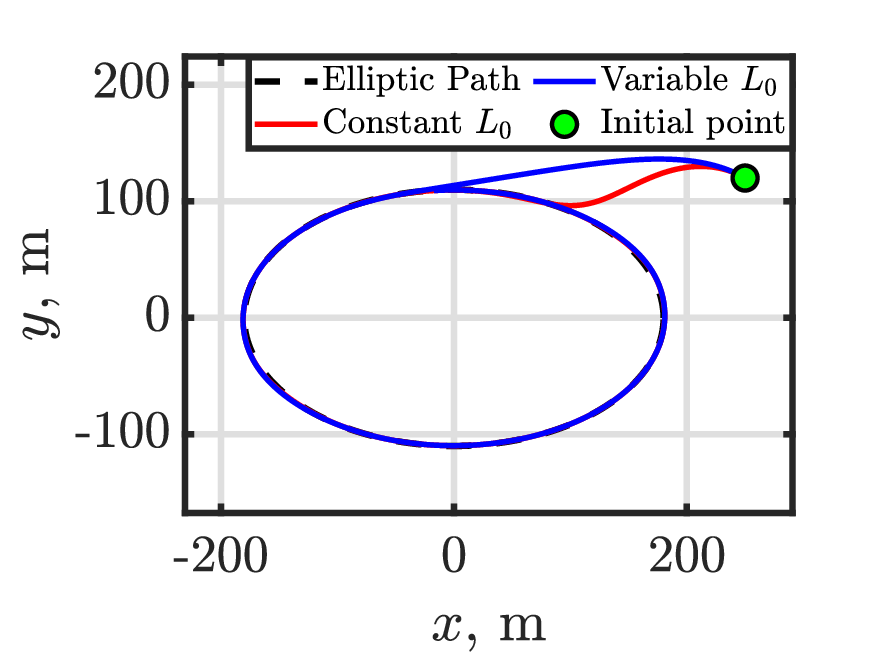}
	\caption{ UAV trajectory}		\label{fig:  UAV trajectory ellipse}       
\end{subfigure}%
\begin{subfigure}[b]{.25\textwidth}
	\centering			\includegraphics[width=\linewidth]{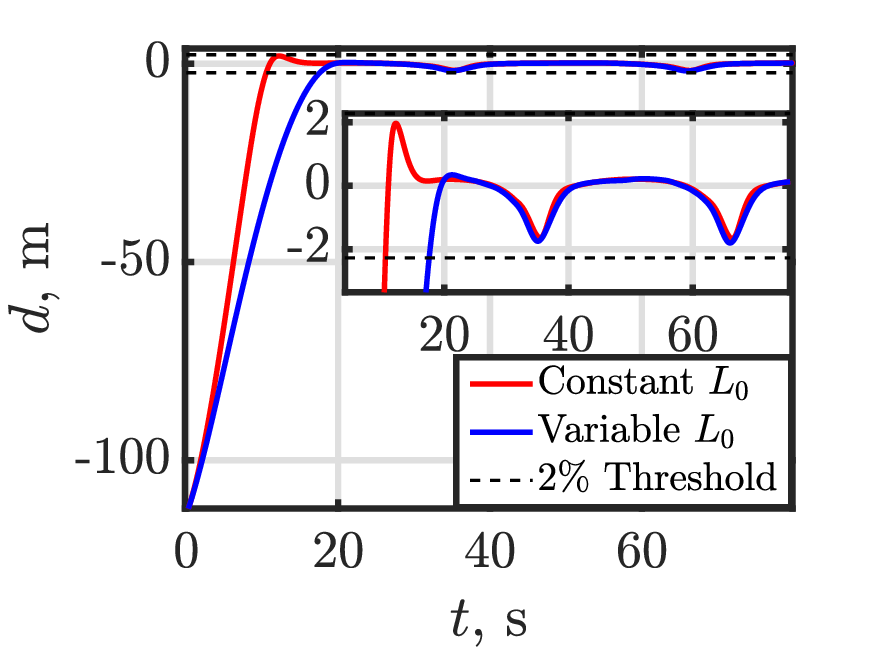}
	\caption{Tracking error}         \label{fig: Tracking error ellipse}
\end{subfigure}%
\qquad
\begin{subfigure}[b]{.25\textwidth}
	\centering			\includegraphics[width=\linewidth]{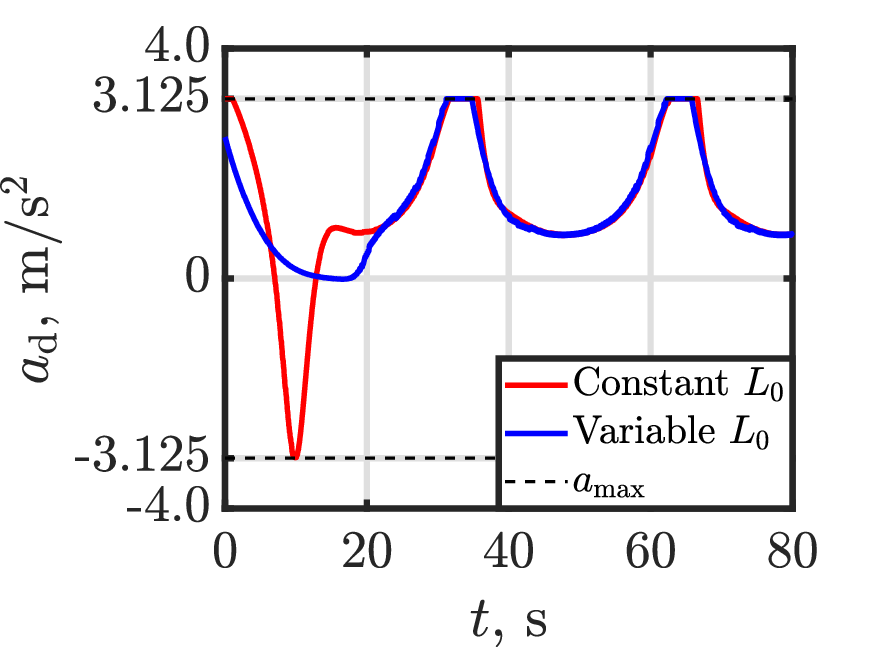}
	\caption{Lateral acceleration}         \label{fig: Lateral acceleration ellipse}
\end{subfigure}%
\begin{subfigure}[b]{.25\textwidth}
	\centering			\includegraphics[width=\linewidth]{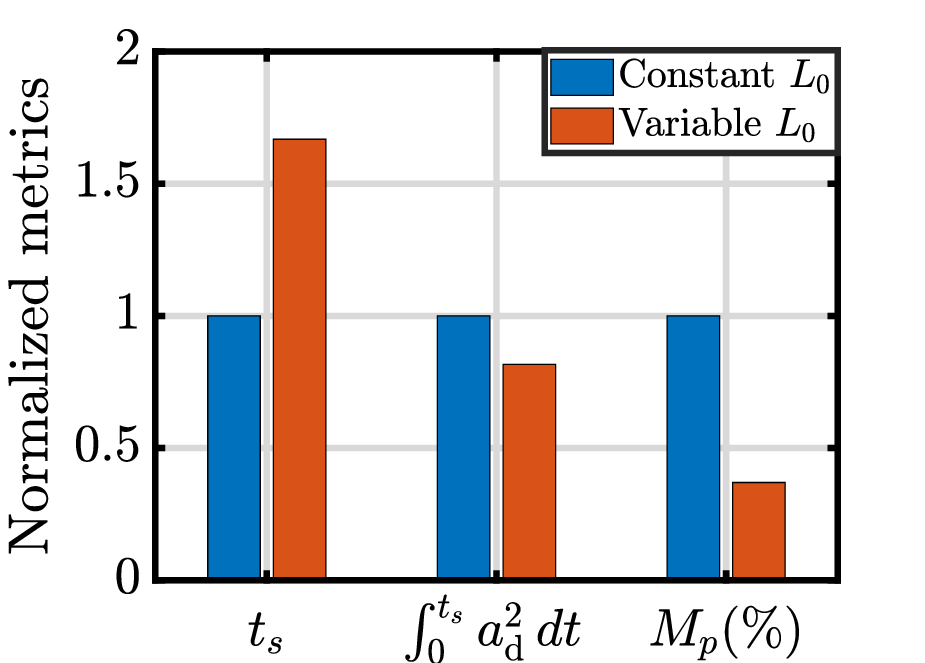}
	\caption{Performance comparison}         \label{fig: Performance comparison ellipse}
\end{subfigure}%
\caption{Results for elliptic path following}
\label{fig: Results for elliptic path following}
\end{figure}
For an elliptic path following case study, the desired path is charactersied by $\frac{x^2}{180^2} + \frac{110^2}{b^2} = 1$. Simulation results are illustrated in Fig.~\ref{fig: Results for elliptic path following}. Fig.~\ref{fig:  UAV trajectory ellipse} shows the UAV trajectory, and the tracking error profile is depicted in Fig.~\ref{fig: Tracking error ellipse}. As seen in Fig.~\ref{fig: Tracking error ellipse}, the variable $L_0$ guidance reduces the tracking error to negligibly low value and exhibits significantly smaller overshoot. This improvement is achieved while maintaining a smoother lateral acceleration profile, as shown in Fig.~\ref{fig: Lateral acceleration ellipse}, where the constant $L_0$ case produces larger and more frequent acceleration peaks due to its inability to adapt the look-ahead distance. The performance metrics in Fig.~\ref{fig: Performance comparison ellipse} establish that the variable $L_0$ guidance yields reduced control effort and overshoot, demonstrating that adaptively increasing the look-ahead distance for large initial deviations effectively prevents saturation and improves transient response. 
\vspace{-10.2pt}
\section{Conclusions}\label{sec:Conclusions}
This work presents an adaptive look-ahead guidance strategy that enhances UAV path-following performance by exploiting the geometric relationship between tracking errors and the look-ahead distance. By characterizing the saturated and unsaturated regions in the $(d,\eta)$ error space, the proposed variable $L_0$ 
method expands the unsaturated operating envelope compared with the constant $L_0$ baseline. A key contribution of this study is the ratio analysis of $L_\mathrm{max}/L_\mathrm{min}$, which quantifies how the envelope expands as this ratio increases. The results show that for ratios greater than 3, the unsaturated region grows by more than 70\%, enabling the guidance law to remain unsaturated over a significantly broader range of tracking errors. This enlarged operational region directly contributes to smoother trajectory corrections, earlier exit from saturation, and reduced overall control effort.
 Simulation studies on straight-line and elliptic paths confirm that dynamically adjusting the look-ahead distance yields smooth converging behaviour, reduced overshoot, and a substantial reduction in control effort. These improvements demonstrate the practical significance of incorporating adaptive look-ahead mechanisms into UAV guidance design, offering a computationally efficient approach that enhances maneuvering efficiency and robustness under realistic actuator and curvature constraints paths.
  These insights also open research directions for extending variable-parameter guidance laws to more complex path geometries and higher-order vehicle models.
\vspace{-10.2pt}
\bibliography{PF_ECC24}

\end{document}